\documentclass[acmsmall,nonacm]{acmart}

\makeatletter
\def\@ACM@checkaffil{%
  \if@ACM@instpresent\else
    \ClassWarningNoLine{\@classname}{No institution present for an affiliation}%
  \fi
  \if@ACM@citypresent\else
    \ClassWarningNoLine{\@classname}{No city present for an affiliation}%
  \fi
  \if@ACM@countrypresent\else
    \ClassWarningNoLine{\@classname}{No country present for an affiliation}%
  \fi
}
\makeatother

\usepackage{amsmath}
\usepackage{amsfonts}
\usepackage{relsize}
\usepackage{mathtools}
\usepackage{xcolor}
\usepackage{xspace}

\usepackage{enumerate}
\usepackage{algorithm}
\usepackage{algpseudocodex}
\usepackage{paralist}
\usepackage{booktabs}
\usepackage{cleveref}
\usepackage{hyperref}

\fancypagestyle{titlepage}{
    \fancyhf{} %
}

\AtBeginDocument{%
  }

\setcopyright{acmlicensed}
\copyrightyear{2018}
\acmYear{2018}
\acmDOI{XXXXXXX.XXXXXXX}
\acmConference[Conference acronym 'XX]{Make sure to enter the correct
  conference title from your rights confirmation emai}{June 03--05,
  2018}{Woodstock, NY}
\acmISBN{978-1-4503-XXXX-X/18/06}

\begin{document}

\newcommand{\fixme}[1]{{\color{red}\ensuremath{\ll}FIXME: #1\ensuremath{\gg}}}

\newcommand{\prot}{\textsc{Kudzu}}

\newcommand{\abs}[1]{\lvert #1 \rvert}
\newcommand{\ceil}[1]{\lceil #1 \rceil}
\newcommand{\Int}[1]{[#1]}
\newcommand{\var}[1]{\textit{#1}}
\newcommand{\deq}{\coloneqq}
\newcommand{\lit}[1]{\textsf{#1}}
\newcommand{\defn}[1]{\textbf{#1}}
\newcommand{\proc}[1]{\textsf{#1}}
\newcommand{\RM}[1]{_{\text{\rm #1}}}
\newcommand{\true}{\mathsf{true}}
\newcommand{\false}{\mathsf{false}}

\newcommand{\arxiv}[1]{{arXiv:#1, \url{http://arxiv.org/abs/#1}}}

\algrenewcommand\algorithmicloop{\textbf{while} $\lnot \textit{done}$ \textbf{wait until either}}

\newcommand{\Pref}[1]{Protocol~\ref{#1}}

\title{\prot: Fast and Simple High-Throughput BFT}

\author{Victor Shoup}
\email{victor@shoup.net}
\affiliation{%
  \institution{Offchain Labs}
}

\author{Jakub Sliwinski}
\email{kobi@anza.xyz}
\affiliation{%
  \institution{Anza}
}

\author{Yann Vonlanthen}
\email{yvonlanthen@ethz.ch}
\affiliation{%
  \institution{ETH Zurich}
}

\sloppy %

\begin{abstract}
We present \prot, a high-throughput atomic broadcast protocol with an integrated fast path. 
Our contribution is based on the combination of two lines of work. 
Firstly, our protocol achieves finality in  just two rounds of communication if all but $p$ out of $n = 3f + 2p + 1$ participating replicas behave correctly, where $f$ is the number of Byzantine faults that are tolerated. 
Due to the seamless integration of the fast path, even in the presence of more than $p$ faults, our protocol maintains state-of-the-art characteristics. 
Secondly, our protocol utilizes the bandwidth of participating replicas in a balanced way, alleviating the bottleneck at the leader, and thus enabling high throughput. 
This is achieved by disseminating blocks using erasure codes. 
Despite combining a novel set of advantages, \prot\ is remarkably simple: 
intricacies such as ``progress certificates'', complex view changes, and speculative execution are avoided.
\end{abstract}
\maketitle

\pagenumbering{arabic}

\floatname{algorithm}{Protocol}

\section{Introduction}

\label{sec:introduction}
Recent years have seen a remarkable surge in popularity and development of resilient distributed systems. The area of blockchain has become a hotbed of research, where systems akin to decentralized world-computers~\cite{buterin2013ethereum,solana,blackshear2023sui} compete to introduce ever-improving protocols. 
At the heart of any such system is a {\em atomic broadcast} protocol \cite{cachin2001secure},
which allows all replicas in the network to agree on a stream of transactions.

The crucial requirement that these protocols must satisfy is {\em Byzantine fault tolerance (BFT)}~\cite{lamport1982byzantine}, which is the ability of a system composed of $n$ different replicas to continue to function even if some of the replicas fail in arbitrary (potentially adversarial) ways. 
Moreover, the network connecting the replicas can be unreliable, or even controlled by an attacker. 
These harsh conditions meant that early global decentralized systems suffered severe disadvantages compared to centralized counterparts, hindering adoption. Despite the challenging setting, research has continued to improve once infamously slow decentralized protocols, and decentralized systems are closing the performance gap.

One key dimension of atomic broadcast performance is the {\em finalization latency} of new transactions. 
Historically, the protocols with the best finalization latency are 
protocols, such as PBFT \cite{PBFT-short}, where a designated leader proposes a block of transactions,
and which work in the partially synchronous communication model \cite{PartialSynchrony}.
In this model liveness is only guaranteed during periods of time where the network is well behaved,
but correctness (i.e., safety) is guaranteed unconditionally.
For such protocols, it is natural to measure finalization latency
as the amount of time that may elapse between when the leader proposes a block
and when all other replicas finalize that block.
In measuring finalization latency, we assume the leader is honest and the network
is well behaved.
Such protocols assume 
 $n \ge 3f + 1$, 
 where $f$ is a bound on the number of corrupt replicas,
 and achieve finalization latency as low as $3\delta$, 
 where $\delta$ is the longest actual message latency between the replicas of the system. 
However, even $2\delta$ finalization latency is possible in situations 
where no more than $p$ replicas fail and $n \ge 3f+2p+1$ for $p \ge 0$ (or even $n \ge 3f + 2p - 1$ for  $p \ge 1$).
Protocols achieving $2\delta$ finalization latency in these circumstances employ a special {\em fast path}.
This type of protocol was first explored in \cite{martin2006fast}.
Ideally, such protocols would still maintain a finalization latency of $3\delta$,
even if more than $p$ replicas fail,
by running a traditional $3\delta$ slow path alongside the fast path.
However, not all fast path protocols enjoy this property
(in particular, the protocol in \cite{martin2006fast} does not).
Moreover,  fast path protocols typically suffer from added complexity, such 
as complex view-change logic, ``progress certificate'' messages,
or speculative execution logic.
This added complexity has led to a history of errors
(indeed, as pointed out in \cite{abraham2017revisiting},
the protocol in \cite{martin2006fast} has a liveness bug).

Another crucial dimension of atomic broadcast
performance is its {\em throughput}, that is, the number of transactions that the protocol can process over time given the fixed bandwidth available at every replica. 
Unfortunatlely, as has been observed and reported in several works \cite{MirBFT,HoneyBadger,Narwhal},
leader-based protocols often suffer from a severe ``bandwidth bottleneck''
at the leader.
However, this leader bottleneck can be easily eliminated
while still maintaining the leader-based structure and all of its practical advantages \cite{song}.
This is done by using erasure codes to ensure that the leader can
disseminate large blocks with {\em low} and {\em well-balanced} communication complexity.

Another desirable feature of leader-based protocols
is lightweight view-change logic that suppports frequent leader rotation.
In older protocols, such as PBFT, a leader is generally kept 
in place for an extended period of time,
and a complex and somewhat inefficient 
view-change subprotocol is used to switch to a new leader if the current leader is
suspected of being faulty by other replicas.
A newer breed of protocols, typified by HotStuff \cite{HotStuff},
employ extremely lightweight view-change logic that supports frequent leader rotation.
Frequently rotating leaders can be beneficial for multiple reasons, for instance, to increase fairness when block production comes with rewards (e.g., maximal extractable value), or to increase censorship resistance.

In this work, our aim is to unite and improve the state-of-the-art in these key dimensions, and to do so while keeping the protocol as simple as possible. 

\textbf{Our contribution.} We present \prot\footnote{Kudzu is an insidious, fast-growing vine, also known as Mile-a-Minute.}: a fast and high-throughput atomic broadcast protocol that is remarkably simple compared to its predecessors. \prot\ is the first BFT protocol that combines an optimistic $2 \delta$ latency fast path integrated into a $3 \delta$ latency slow path, with high-throughput data dispersal, as well as lightweight view change logic that supports frequent leader rotation.
We provide a detailed description of the protocol and rigorously prove its security.

\begin{enumerate}
    \item \textbf{Fast path.} In a network of $n = 3f + 2p +1$ replicas,  \prot\ achieves finalization latency of  $2\delta$ if at most $p$ replicas are corrupt.
    \item \textbf{Simultaneous slow path.} If more than $p$ replicas fail, \prot\ maintains the best possible finalization latency of $3\delta$ by running a slow path alongside the fast path. 
    
    \item \textbf{High throughput.} High throughput is achieved by \prot\ 
    using erasure codes to ensure that the leader can
disseminate large blocks with low and well-balanced communication complexity.
    
    \item \textbf{Lightweight view change.}  
    \prot\ employs an extremely simple and efficient view change logic that allows for frequent leader rotation.
\end{enumerate}
\prot\ is the first atomic broadcast protocol to satisfy all of these properties simultaneously.
It also enjoys other properties, such as optimistic responsiveness (the protocol
proceeds as fast as the network will allow, with no artificially introduced delays),
and a block time (the delay between successive honest leaders proposing
a block) of just $2\delta$.

\textbf{Technical Intuition.} 
Inspired by DispersedSimplex~\cite{song}, \prot\ introduces the minimal changes to that protocol required to incorporate a fast path. 
A designated leader in \prot\ distributes the block data by sending erasure coded fragments to all other replicas. In turn, these replicas then broadcast the fragments themselves, together with a first-round vote on the cryptographic hash identifying the block. The voting logic carefully incorporates a rule that makes sure that a block can already be finalized given $n-p$ first-round votes. As a result, replicas can already reassemble the block and finalize it in $2\delta$ latency in the good case where the network is well behaved,
the leader is honest, and at most $p$ replicas are corrupt.
Another vote-counting rule is used to finalize the block on a slower, $3\delta$-latency path,
if more than $p$ replicas are corrupt.
To avoid getting stuck altogether, a replica may also vote for a couple of other blocks than the one it initially 
received from the leader, including a special ``timeout'' block.
The main innovation of our paper is the
logic for determining when these ``extra'' votes are cast ---
this logic is deceptively  simple, but carefully
 maintains a very delicate balance between liveness and safety.
 With this innovation, 
 we avoid intricacies such as ``progress certificates'', complex view changes, and speculative execution 
 as found in other protocols (such as Banyan~\cite{banyan}, SBFT~\cite{gueta2019sbft}, Kuznetsov et al.~\cite{kuznetsov2021revisiting}, and HotStuff-1~\cite{kang2024hotstuff}).

\section{Related Work}
\label{sec:relatedwork}

\subsection{Fast Path Protocols}
\label{sec:relatedwork:fast}
A long line of work proposes consensus protocols with a fast path, typically called ``fast'', ``early-stopping'' or ``one-step'' consensus~\cite{brasileiro2001consensus,kursawe2002optimistic,dutta2004complexity,song2008bosco,martin2006fast,guerraoui2007refined}. 
FaB Paxos~\cite{martin2006fast} introduces a parametrized model with $3f + 2p + 1$ replicas, where $p \ge 0$. The parameter $p$ describes the number of replicas that are not needed for the fast path.
These protocols can terminate optimally fast in theory ($2\delta$, or 2 network delays) under optimistic conditions.
The papers \cite{dutta2004complexity,kuznetsov2021revisiting,abraham2021good} point out that the lower bound of $3f + 2p + 1$ actually only applies to a restricted type of protocol. The papers present single-shot consensus protocols that use only $3f + 2p^* -1$ replicas, with $p^* \geq 1$, and prove the corresponding lower bound.

\subsection{Simultaneous Slow Path Protocols}
\label{sec:relatedwork:integratedslow}
Interestingly, in practice, these fast path protocols might {\em increase} the finalization latency, as the fast path requires a round of voting between $n-p$ replicas, which could be slower than two rounds of voting between $n-f-p$ replicas that are concentrated in a geographic area. Banyan~\cite{banyan} performs the fast path in parallel with the $3\delta$ mechanism, which is optimally fast if more than $p$ replicas are faulty or exhibit higher network latency. However, Banyan can exhibit \textit{unbounded} message complexity when there is a corrupt leader.
\prot\ addresses this drawback, shares the same optimistic latency properties, improves throughput through balanced dispersal, benefits from a simpler design, and better properties during leader rotation (see \cref{sec:relatedwork:viewchange}).

\subsection{High-Throughput Protocols}
\label{sec:relatedwork:bottlenecks}
Leader-based protocols such as PBFT~\cite{PBFT-short}, HotStuff~\cite{HotStuff} and Tendermint~\cite{Tendermint} suffer from a bandwidth bottleneck, as in these protocols the leader is responsible for disseminating transactions to all replicas.
As mentioned above,
this bottleneck has been well known for quite a while \cite{MirBFT,HoneyBadger,Narwhal}.
One way to alleviate this bottleneck is to move away from leader-based protocols
to a more symmetric, leaderless protocol design where all replicas disseminate transactions.
Such an approach was already taken in \cite{MirBFT,HoneyBadger,DagRider,Narwhal}.
These protocols typically have worse latency than leader-based protocols,
and moreover, since many replicas may end up broadcasting the same transactions,
the supposed improvement in throughput can end up being illusory 
(note that \cite{MirBFT} actually tackles this duplication problem head on).

\prot\ follows a different approach to address the leader bottleneck, which makes use of erasure coding~\cite{solana,turbine}. By splitting blocks into smaller, erasure-coded shares, the leader can transmit less data, leading to a balanced utilization of resources.

This line of work shines by providing high throughput and low latency~\cite{song,pipes}. Alpenglow~\cite{alpenglow} is a recent proof-of-stake protocol being deployed for the Solana blockchain that showcases the appeal of this line of work in practice. Alpenglow and \prot\ are intimately related in their voting logic.

\subsection{View Changes}
\label{sec:relatedwork:viewchange}

Some atomic broadcast protocols are notorious for having complex view changes, especially in the presence of a fast path~\cite{aublin2015next}. In the case of Zzyzzva~\cite{kotla2007zyzzyva} and UpRight~\cite{clement2009upright}, safety errors were later pointed out~\cite{abraham2017revisiting}. Arguably SBFT~\cite{gueta2019sbft} was the first to correct these mistakes. SBFT~\cite{gueta2019sbft} avoids the all-to-all broadcast
that occurs in \prot,
but this results in higher finalization latency. Moreover, SBFT does not address the above-mentioned leader bottleneck.

\section{Model and Preliminaries}

We consider a network of $n$ replicas $P_1, \dots, P_n$ called replicas. Up to $f$ replicas can be Byzantine, i.e., deviate from the protocol in arbitrary ways, such as collude to attack the protocol. The remaining replicas follow the protocol and are referred to as honest. We aim to provide better latency if only up to $p$ replicas do not cooperate. In other words, $n-p$ honest replicas including the leader are enough for the fast path to be effective.

To prove the security of \prot\, we assume
\begin{equation}
\label{eq:n_lower}
    n \ge 3f + 2p +1,
\end{equation}
where $f \ge 1$ and $p \ge 0$.
Moreover, to prove concrete bounds on message, communication, and storage complexity,
we assume
\begin{equation}
\label{eq:n_upper}
    n < 3 (f + p + 1).
\end{equation}
Assumption (\ref{eq:n_upper}) is not a real restriction.
Indeed, if $n \ge 3(f+p+1) = 3f + 3p + 3$, then we can always
{\em increase} $p$ appropriately, leaving $n$ and $f$ alone, so that both (\ref{eq:n_lower})
and (\ref{eq:n_upper}) are satisfied.
This only increases the overall performance of the protocol.

\subsection{Network Assumptions}

We will not generally assume network synchrony.
However,
we say the network is {\em $\delta$-synchronous at time $T$}
if every message sent from an honest replica $P$ at or before time $T$
to an honest replica $Q$ 
is received by $Q$ before time $T+\delta$.
We also say the the network is \defn{$\delta$-synchronous over an interval
$[a,b+\delta]$} if it is $\delta$-synchronous at time $T$
for all $T \in [a,b]$.

While our protocols always guarantee {\em safety},
even in periods of asynchrony,
{\em liveness} will only be guaranteed in periods
of $\delta$-synchrony, for appropriately bounded $\delta$ (and this synchrony bound may be explicitly used by the protocol).
Thus, we are essentially working in the {\em partial synchrony} model of \cite{PartialSynchrony}.
However, instead of assuming a single point in time (GST)
after which the network is assumed to be synchronous,
we take a somewhat more general point of view that models
a network that may alternate between periods of asychrony and synchrony.

Replicas have local clocks that can measure the passage of local time.
We do not assume that clocks are synchronized in any way.
However, we do assume that there is no {\em clock skew}, that is, all clocks tick at the same rate (but we could also 
just assume the skew is bounded and incorporate that 
bound into the protocol's synchrony bound).

\subsection{Problem Statement}

The purpose of state machine replication is to totally order blocks containing transactions, so that all replicas output transactions in the same order. Our protocol orders blocks by associating them with natural numbered slots. Some leader replica is assigned to every slot. For every slot, either some block produced by the leader might be finalized, or the protocol can yield an empty block. The guarantees of our protocol can be stated as follows:

\begin{itemize}
\item \textbf{Safety}: If some honest replica finalizes block $B$ in slot $v$, and another honest replica finalizes block $B'$ in slot $v$, then $B = B'$.
\item \textbf{Liveness}: If the network is in a period of synchrony, each honest replica continues to finalize blocks for slots $v = 1,2,\dots$ 
\end{itemize}

\noindent In addition to liveness and safety, we support a fast path:

\begin{itemize}
    \item \textbf{Fast Termination}: If the network is in a period of synchrony, $n - p$ replicas behave momentarily honestly, and an honest leader proposes a block $B$ at time $t$, then every honest replica finalizes $B$ at time $t + 2\delta$.
\end{itemize}

\subsection{Cryptographic Assumptions}

\subsubsection{Signatures and certificates.}
\label{sec:sigs}

We make standard cryptographic assumptions of secure digital signatures and collision-resistant hash functions. 
We assume all replicas know the public keys of other replicas. 

We use a $k$-out-of-$n$ threshold signature scheme.
We refer to a {\em signature share} and a {\em signature certificate}:
signature shares from $k$ replicas on a given message
may be combined to form a signature certificate on that message.
This can be implemented in various ways, e.g., 
based on BLS signatures~\cite{BLS-sigs,cryptoeprint:2018/483,DBLP:conf/pkc/Boldyreva03}.
The security property for such a threshold signature scheme
may be stated as follows.
\begin{description}
\item[Quorum Size Property:]
It is infeasible to produce a signature certificate 
on a message $m$,
unless $k-f'$ honest replicas have 
issued signature shares on $m$, where $f' \le f$ 
is the number of corrupt replicas.
\end{description}
For ease of exposition and analysis, we assume {\em static corruptions}, so the
adversary must choose some number $f' \le f$ replicas to corrupt at the
very beginning of the protocol execution, and then does not corrupt any
replicas thereafter.
That said, we believe all of our the protocols are secure against
{\em adaptive corrupts}, provided the threshold signature scheme is as well.

As we will see, we will need a one $k$-out-of-$n$ threshold signature scheme
with $k=n-f-p$, and another with $k=n-p$.

\subsubsection{Information Dispersal.}
\label{sec:avid}

We make use of well-known techniques for {\em asynchronous verifiable information dispersal (AVID)} involving erasure codes and Merkle trees~\cite{avid}.

\paragraph{Erasure codes.}
For integer parameters $k \ge d \ge 1$, a {\em $(k, d)$-erasure code} encodes a bit string $M$ as a vector
of $k$ {\em fragments}, $f_1, \ldots, f_k$, in such a way that 
any $d$ such fragments may be used to efficiently reconstruct $M$.
Note that for variable-length $M$, 
the reconstruction algorithm
also takes as input the length $\beta$ of $M$.
The reconstruction algorithm may fail 
(for example, a formatting error)---if it fails it returns $\bot$, while if it succeeds
it returns a message that when re-encoded
will yield $k$ fragments that agree with the original
subset of $d$ fragments. 
We assume that all fragments have the same size,
which is determined as a function of $k$, $d$, and $\beta$.

Using a Reed-Solomon code, which is based on polynomial
interpolation, we can realize
a  $(k, d)$-erasure code so that
if $\abs{M} = \beta$, then each fragment has size
$\approx \beta/d$.

In our protocol,
the payload of a block will be encoded using an 
$(n, f+p+1)$-erasure code.
Such an erasure code encodes a payload $M$ as a vector of fragments
$f_1, \ldots, f_n$, any $f+p+1$ of  which can be used to 
reconstruct $M$.
This leads to a data expansion rate
of at most roughly $3$;
that is, $\sum_i \abs{f_i} \approx n/(f+p+1) \cdot \abs{M} < 3 \abs{M}$,
where last inequality follows from assumption (\ref{eq:n_upper}).

\paragraph{Merkle trees.}
Recall that a Merkle tree allows one replica $P$ to
commit to a vector of values $(v_1, \ldots, v_k)$
using a collision-resistant hash function
by building a (full) 
binary tree whose leaves are the hashes of $v_1, \ldots, v_k$,
and where each internal node of the tree is the hash of its 
two children.
The root $r$ of the tree is the commitment.
Replica $P$ may ``open'' the commitment at a position $i \in \Int{k}$
by revealing $v_i$ along with a ``validation path'' $\pi_i$,
which consists of the siblings of all nodes along the 
path in the tree from the hash of $v_i$ to the root $r$. 
We call $\pi_i$ a {\em validation path from the root under $r$
to the value $v_i$ at position $i$}.
Such a validation path is checked by recomputing the
nodes along the corresponding path in the tree,
and verifying that the recomputed root is equal to the given commitment $r$. 
The collision resistance of the hash function ensures that
$P$ cannot open the commitment
to two different values at a  given position.

\paragraph{Encoding and decoding.}
For a given payload $M$ of length $\beta$, we will encode
$M$ as a vector of fragments $(f_1, \ldots, f_n)$
using an $(n, f+p+1)$-erasure code,
and then form a Merkle tree with root $r$ whose leaves are the hashes
of $f_1, \ldots, f_n$.
We define the {\em tag $\tau \deq (\beta,r)$}.

For a tag $\tau=(\beta,r)$,
we shall call $(f_i, \pi_i)$ a 
{\em certified fragment for $\tau$ at position $i$} if
\begin{itemize}
\item
$f_i$ has the correct length of a fragment for a message of length $\beta$, 
and
\item
$\pi_i$ is a correct 
validation path from the root under $r$
to the fragment $f_i$ at position $i$.
\end{itemize}

The function $\var{Encode}$ takes as input a payload $M$. 
It builds a Merkle tree for $M$ as above with root $r$
(encoding $M$ as a vector of fragments, 
and then building the
Merkle tree whose leaves are the hashes of all of these fragments).
It returns
\[
\big(\ \tau, \ \{ (f_i,\pi_i) \}_{i \in \Int{n}} \ \big),
\]
where $\tau$ is the tag $(\beta,r)$, $\beta$ is the length of $M$, and
each 
$(f_i, \pi_i)$ is a 
certified fragment for $\tau$ at position $i$.

The function $\var{Decode}$ takes as input
\[
\big(\ \tau, \ \{ (f_i,\pi_i) \}_{i \in \mathcal{I}} \ \big),
\]
where $\tau=(\beta,r)$ is a tag,
$\mathcal{I}$ is a subset of $\Int{n}$ of size $f+p+1$,
and each 
$(f_i, \pi_i)$ is a 
certified fragment for $\tau$ at position $i$.
It first reconstructs a message $M'$ 
from the fragments $\{f_i\}_{i \in \mathcal{I} }$,
using the size parameter $\beta$.
If $M' = \bot$, it returns $\bot$.
Otherwise, it encodes $M'$ as a vector of fragments $(f'_1,\ldots,f'_n)$
and Merkle tree with root $r'$ from $(f'_1,\ldots,f'_n)$.
If $r' \ne r$, it returns $\bot$.
Otherwise, it returns $M'$.

Under collision resistance
for the hash function used for the Merkle trees, 
any $f+p+1$ certified fragments for given tag $\tau$ 
will decode to the same payload ---
moreover, if $\tau$ is the output of the encoding function,
these fragments will decode to $M$
(and therefore, if the decoding function outputs $\bot$, 
we can be sure that $\tau$ was maliciously constructed).
This observation is the basis for the 
protocols in  
\cite{SodsBC,DBLP:conf/podc/LuL0W20,yang2022dispersedledger}.
Moreover, with this approach, we do not need to use
anything like an ``erasure code proof system'' 
(as in \cite{cryptoeprint:2021/1500}),
which would add significant computational complexity
(and in particular, the erasure coding would have to be done
using parameters compatible with the proof system,
which would likely lead to much less efficient 
encoding and decoding algorithms).

\section{\prot\ Protocol}\label{simple_bound}

\prot\ iterates through slots, where in each slot there is a designated leader who proposes a new block, which is chained to a parent block. Leaders may be rotated in each slot, either in a round-robin fashion or using some pseudo-random sequence. The slot leader disseminates large blocks in a way that keeps the overall communication complexity low and avoids a bandwidth bottleneck at the leader. The communication is balanced, meaning that each replica, including the leader, transmits roughly the same amount of data over the network.

We describe our protocol as a few simple subprotocols that run concurrently with each other:
\begin{itemize}
    \item Vote and Certificate Pool: data structure managing the votes and certificates;
    \item Complete Block Tree: data structure storing the reconstructed blocks;
    \item Main Loop: loop issuing votes that makes sure some blocks become \emph{notarized} and \emph{finalized}. Notarized blocks can be reconstructed by all replicas and are added to the Complete Block Tree. Finalized blocks are ordered and output by the protocol.
\end{itemize}

\subsection{Protocol Data Objects}

\begin{definition}[block]
A \defn{block $B$} is of the form $\lit{Block}(v, \tau, h_p)$, where
\begin{enumerate}
\item
$v \in \{1, 2, \ldots\}$ is the slot number associated with the block
(and we say {\em $B$ is a block for slot $v$}),
\item
$\tau$ is the tag obtained by encoding $B$'s payload $M$,
\item
$h_p$ is the hash of $B$'s parent block 
(or $h_p=\bot$ by convention if $B$'s parent is a notional ``genesis'' block).
\end{enumerate}
We also call a certified fragment for the tag $\tau$ a {\em certified fragment for $B$}.
The block $B_v^{\lit{timeout}} = \lit{Block}(v, \bot, \bot)$ is a special \emph{timeout block}. %
\end{definition}

\begin{definition}[votes and certificates]
A \defn{notarization vote from $P_i$ for block $B$} is an
object of the form $\lit{NotarVote}(B, \sigma_i, f_i, \pi_i)$,
where $\sigma_i$ is a valid signature share from $P_i$ on the object $\lit{Notar}(B)$,
and $(f_i,\pi_i)$ is either a certified fragment for $B$ at position $i$, or $(f_i,\pi_i) = (\bot, \bot)$ if $B = B_v^{\lit{timeout}}$.

A \defn{notarization certificate for $B$} is an
object of the form $\lit{NotarCert}(B, \sigma)$,
where $\sigma$ is a valid $(n-f-p)$-out-of-$n$ signature certificate on the object $\lit{Notar}(B)$. The notarization vote on the timeout block is also called the timeout vote, and the notarization certificate for the timeout block is called the timeout certificate.

A \defn{first vote from $P_i$ on block $B$} 
is an
object of the form $\lit{FirstVote}(\sigma_i' ,\lit{NotarVote}(B, \sigma_i, f_i, \pi_i))$,
where $\sigma_i'$ is a valid signature share from $P_i$ on the object $\lit{First}(B)$,
and $\lit{NotarVote}(B, \sigma_i, f_i, \pi_i)$ is a notarization vote from $P_i$ on block $B$. %

A \defn{fast finalization certificate for $B \neq B_v^{\lit{timeout}}$} is an
object of the form $\lit{FirstCert}(B, \sigma)$,
where $\sigma$ is a valid $(n-p)$-out-of-$n$ signature certificate on the object $\lit{First}(B)$.

A \defn{finalization vote from $P_i$ on block $B$} is an
object of the form $\lit{FinalVote}(B, \sigma_i)$,
where $\sigma_i$ is a valid signature share from $P_i$ on the object $\lit{Final}(B)$.

A \defn{finalization certificate for $B$} is an
object of the form $\lit{FinalCert}(B, \sigma)$,
where $\sigma$ is a valid $(n-f-p)$-out-of-$n$ signature certificate on the object $\lit{Final}(B)$.
\end{definition}

\subsection{Vote and Certificate Pool}\label{sec:pool}

Each replica maintains a {\em pool} with votes and certificates. For every slot, the pool stores votes and certificates associated with the slot.

As we will see, by design, for any one slot, a honest replica can send only 1 first vote, 1 timeout vote,
and 1 finalization vote.
As we will also see later (in \Cref{sec:bounds}), 
for one slot a honest replica can only send at most 3 (non-timeout) notarization votes.
Any votes exceeding this bound can only result from misbehavior and are not added to the pool. For example, if a replica receives more than three notarization votes for a given slot from some replica $P$, the replica can ignore these votes and conclude that $P$ is corrupt. In particular, only one first vote per replica can be observed by the protocol loop in \Pref{alg:main_loop_v3}. When a first vote is added to the pool, also the contained notarization vote is added to the pool.

Whenever a replica receives enough votes, and it does not already have a corresponding certificate,
it will generate the certificate, add it to the pool,
and broadcast the certificate to all replicas.
Similarly, whenever a replica receives a
certificate, and it does not already have a corresponding certificate, 
it will add it to the pool,
and broadcast the certificate to all replicas.
For one slot it is impossible that the pool would receive or create more than: 
1 timeout certificate, 1 fast finalization certificate, 1 finalization certificate,
and 5 notarization certificates.
The first bound is immediate, since there can be only one timeout certificate per slot.
The second and third bounds follow from the safety analysis below.
The fourth bound follows from the analysis in \Cref{sec:bounds}.

\subsection{Complete Block Tree}
\label{sec-block-pool}
Each replica also maintains a {\em complete block tree},
which is a tree of blocks
rooted at a notional genesis block at slot 0.
We will show that the number of blocks for a given slot is bounded by 5.
A block $B=\lit{Block}(v,\tau, h_p)$  is added to the tree if 
each of the following holds:
\begin{itemize}
\item
the certificate pool contains a notarization certificate for $B$ and $B \neq B_v^{\lit{timeout}}$;
\item
$h_p=\bot$ or 
the complete block tree 
contains a parent block with the hash $h_p$;
\item
the replica has received enough (i.e. $f+p+1$) notarization votes to reconstruct the effective payload $M$ of $B$
as 
\[
M \gets \var{Decode}(\tau, \{ (f_i,\pi_i) \}_{i\in\mathcal{I}}),
\]
where $\{ (f_i,\pi_i) \}_{i\in\mathcal{I}}$ is the corresponding 
collection of certified fragments for $\tau$;
\item
$M \ne \bot$ and satisfies some correctness predicate that may depend
on the path of blocks (and their payloads) from genesis to block $B$.
\end{itemize}
A replica does not broadcast anything in addition to adding a block to the block tree.

\subsection{Finalization}
We say that a block $B$ for slot $v$ is {\em explicitly finalized by replica $P$}
if the complete block tree of $P$ contains $B$
and the certificate pool of $P$ contains either a fast finalization certificate for $B$ or finalization certificate for  $B$.
In this case,
we say that 
all of the predecessors of block $B$
in the complete block tree are {\em implicitly finalized by $P$}.
The payloads of finalized
blocks may be then transmitted in order to the execution layer
of the protocol stack of a replicated state machine.

\subsection{Generating Block Proposals}\label{sec:generating_block}
The logic for generating block proposal material $B$, $(f_1, \pi_1), \ldots, (f_n, \pi_n)$ in slot $v$
in line \ref{line:generate_block} of \Pref{alg:main_loop_v3} is as follows:
\begin{itemize}
\item
build a payload $M$ that validly extends the path in the 
complete block tree ending at a block $B\RM{p}$ with hash $h_p$;  
\item
compute 
\[
(\tau, \{ (f_i,\pi_i) \}_{i \in \Int{n}}) \gets \var{Encode}(M) ; 
\]
\item
set $B \deq \lit{Block}(v,\tau,h_p)$.
\end{itemize}

\subsection{Validating Block Proposals}\label{sec:validating_block}
To check if $\lit{BlockProp}(B, f_j,\pi_j)$ is a valid block proposal in slot $v$ in line~\ref{line:validate_block} of \Pref{alg:main_loop_v3},
replica $P_j$ checks that the following conditions holds: 

\begin{itemize}
\item
the proposal is signed or otherwise authenticated by the leader,
\item
$B$ is of the form $\lit{Block}(v,\tau,h_p)$,
\item
$(f_j,\pi_j)$ is a certified fragment for $\tau$ at position $j$,
\item the complete block tree
contains a block with the hash $h_p$ in a slot $v' < v$;
\item the pool contains timeout certificates for 
slots $v'+1, \ldots, v-1$;
\end{itemize}

These last two conditions might not hold at a given point in time, but may hold at a later point in time, and so might need to be checked again when blocks or certificates are added.

\subsection{Main Loop}

The main protocol for $P_j$ is
described in \Pref{alg:main_loop_v3}.
In the description, $\proc{leader}(v)$ denotes the leader for slot $v$ ---
as mentioned,
leaders may be rotated in each slot,
either in a round-robin fashion or using some pseudo-random sequence.

As mentioned in \Cref{sec:pool}, each replica only considers only one first vote that it receives
from any other replica.
To make this explicit in the protocol, we use a map $\mathit{firstVote}$ from replicas to blocks
to record these votes for the current slot.
The protocol also uses simple helper functions on this map.   
\begin{itemize}
\item $\proc{allVotes}(\textit{firstVote})$: the total number of first votes for slot $v$ contained in the pool,
\item $\proc{maxVotes}(\textit{firstVote})$: the maximal number of first votes on some non-timeout block $B$ for slot $v$ contained in the pool,
\item $\proc{manyVotes}(\textit{firstVote})$: returns the set of non-timeout blocks in slot $v$ on which the pool contains at least $f+p+1$ first votes.
\end{itemize}
For example, if the pool contains 1, 2, 3, 4 first votes on blocks $B_1$, $B_2$, $B_3$, $B_v^{\lit{timeout}}$ respectively (and $f+p+1 = 2$), then $\proc{allVotes} = 10$, $\proc{maxVotes} = 3$, and $\proc{manyVotes} = \{B_2, B_3\}$.

The protocol also uses a subprotocol \proc{ReconstructAndNotarize}$(v,B)$, defined in 
\Pref{alg:reconstruct}.

Each replica $P_j$ moves through slots $v=1,2,\ldots .$
In each slot, it will enter a loop in which it waits for one of several conditions to trigger
an action.
These conditions are based on the objects in its pool and its complete block tree,
as well as local variables.
\begin{itemize}
\item 
Lines \ref{line:exit1_start}--\ref{line:exit1_end} present the logic
for the replica successfully exiting the slot by finding a block $B$ for that slot 
in its complete block tree.
In addition, if the replica did not broadcast a notarization vote for any other block
(including the timeout block) for that slot, it will also broadcast a finalization block for $B$.
\item 
Lines \ref{line:exit2_start}--\ref{line:exit2_end} present the logic
for the replica unsuccessfully exiting the slot by obtaining a timeout certificate for that slot.
\item 
Lines \ref{line:prop_start}--\ref{line:prop_end} present the logic for the replica proposing
a block for that slot if it is the leader for that slot.
It generates the block proposal as in \Cref{sec:generating_block}, extending the path in the
complete block tree ending at $B\RM{p}$.
Here, $B\RM{p}$ is either the genesis block or the block that it found
in its complete block tree the last time it successfully exited a slot
(other choices of $B\RM{p}$ are possible).
\item 
Lines \ref{line:fvblock_start}--\ref{line:fvblock_end} present the logic for the replica broadcasting
a first vote for a non-timeout block $B$.
It will do so only if $B$ is a valid block proposed by the leader 
(as in \Cref{sec:validating_block}) and has not already first voted.
Recall that a first vote for $B$ also includes a notarization vote for $B$.
\item 
Lines \ref{line:time_start}--\ref{line:time_end} present the logic for the replica broadcasting
a first vote for the timeout block for this slot.
It will do so only if a sufficient amount of time has passed since it entered the slot
and has not already first voted.
\item
Lines \ref{line:map_start}--\ref{line:map_end} present the logic for updating the map $\mathit{firstVotes}$.
No other actions are taken.
\item 
Lines \ref{line:second_start}--\ref{line:second_end} present the logic for the replica taking 
a ``second look'' at a block $B$, if it has received sufficiently many first votes for $B$.
It will do so only if it has already first voted (and has not already taken a second look at $B$).
If it can reconstruct a valid payload for $B$, it will broadcast a notarization vote for $B$ (if it has
not already done so);
otherwise, it will broadcast a notarization vote for the timeout block.
\item 
Lines \ref{line:special_start}--\ref{line:special_end} present the logic for the replica 
broadcasting a timeout vote under special circumstances.
It will do so only if it has already first voted and 
\[\proc{allVotes}(\textit{firstVote})-\proc{maxVotes}(\textit{firstVote}) \ge f+p+1.\]
We note that the quantity \[\proc{allVotes}(\textit{firstVote})-\proc{maxVotes}(\textit{firstVote})\] 
cannot decrease as we add entries to $\var{firstVote}$.
That is because, when we add an entry,  the first term increases by 1 and the second either decreases by 1 or remains unchanged. 
\end{itemize}

\begin{algorithm}
\caption{\prot\ main loop for replica $P_j$}\label{alg:main_loop_v3}
\begin{algorithmic}[1]
\State $B\RM{p} \gets \proc{genesis}$ \Comment{parent of the next block}
\For{$v = 1,2,\dots$}
    \State $T\RM{start}      \gets \proc{clock}()$ \Comment{slot‑local initialisation}
    \State $\textit{done}, \textit{proposed}, \textit{firstVoted}    \gets \false$
    \State $\textit{notarized} \gets \{\}$         \Comment{blocks already notarized}
    \State $\textit{secondLook} \gets \{\}$        \Comment{blocks already reconsidered}
    \State $\textit{firstVote}\gets \{\}$          \Comment{map $P_i\mapsto B$ for their first vote}
    \Statex 
    \Loop  \label{line:v3_wait_until}
        \State there exists a block $B$ for slot $v$ in the complete block tree $\Rightarrow$ \label{line:exit1_start}
        \State \>\>\>\> $B\RM{p} \gets B$; $\textit{done}\gets\true$ 
        \State  \>\>\>\> \textbf{if} $\textit{notarized} \subseteq \{B\}$ \textbf{then} broadcast $\lit{FinalVote}(B, \sigma_j)$ \label{line:condition_final_vote} \label{line:exit1_end}

        \Statex  

        \State the pool contains a timeout certificate for $v$ $\Rightarrow$ \label{line:exit2_start}
        \State \>\>\>\> $\textit{done}\gets\true$ \label{line:timeout_cert_done} \label{line:exit2_end}
        \Statex  
        
        \State $\lnot\textit{proposed}\;\land\;\proc{leader}(v)=P_j \Rightarrow$ \label{line:prop_start}
        \State \>\>\>\> $\textit{proposed}\gets\true$
        \State \>\>\>\> generate block proposal material $B$, $(f_1,\pi_1),\dots,(f_n,\pi_n)$ extending block $B\RM{p}$ \label{line:generate_block}
        \State \>\>\>\> \textbf{for} all $i \in [n]$: \textbf{send} $\lit{BlockProp}(B,f_i,\pi_i)$ to $P_i$ \label{line:prop_end} %
        \Statex  

        \State $\lnot\textit{firstVoted}\;\land\;$received valid $\lit{BlockProp}(B,f_j,\pi_j)$ from $\proc{leader}(v)\Rightarrow$\label{line:validate_block} \label{line:fvblock_start}
        \State \>\>\>\> $\textit{firstVoted}\gets\true$
        \State \>\>\>\> broadcast $\lit{FirstVote}(\sigma_j',\lit{NotarVote}(B,\sigma_j,f_j,\pi_j))$ \label{line:notarize_as_first_vote}
        \State \>\>\>\> $\textit{notarized}\gets\textit{notarized}\cup\{B\}$ \label{line:fvblock_end}
        \Statex  

        \State $\lnot\textit{firstVoted}\;\land\;\proc{clock}() > T\RM{start}+\Delta\RM{timeout}\Rightarrow$ \label{line:time_start}
        \State \>\>\>\> $\textit{firstVoted}\gets\true$ \label{line:first_vote_timeout_v3}
        \State \>\>\>\> broadcast $\lit{FirstVote}(\sigma_j',\lit{NotarVote}(B_v^{\lit{timeout}},\sigma_j,\bot,\bot))$
        \State \>\>\>\> $\textit{notarized}\gets\textit{notarized}\cup\{B_v^{\lit{timeout}}\}$ \label{line:time_end}
        \Statex  

        \State received valid $\lit{FirstVote}(\_,\lit{NotarVote}(B,\_))$ from $P_i$ \textbf{and}\; $\textit{firstVote}[P_i] = \bot \Rightarrow$ \label{line:map_start}
        \State \>\>\>\> $\textit{firstVote}[P_i]\gets B$ \label{line:record_first_vote_v3} \label{line:map_end}
        \Statex  

        \State $\textit{firstVoted}\;\land\;\exists B\in\proc{manyVotes}(\textit{firstVote})\setminus\textit{secondLook}$   \label{line:condition_many} \label{line:second_start} 
        \State \>\>\textbf{and}\; $B$'s parent is in the complete block tree $\Rightarrow$
        \State \>\>\>\> $\textit{secondLook}\gets\textit{secondLook}\cup\{B\}$
        \State  \>\>\>\> \proc{ReconstructAndNotarize}$(v,B)$ \label{line:invoke_reconstruct_and_notarize} \label{line:second_end}
        \Statex

        \State $\textit{firstVoted}\;\land\;
               (\proc{allVotes}(\textit{firstVote})-\proc{maxVotes}(\textit{firstVote})\ge f+p+1)$ \label{line:special_start}
        \State \>\>\textbf{and}\; $B_v^{\lit{timeout}}\notin\textit{notarized} \Rightarrow$  \label{line:notarize_timeout_v3}
        \State \>\>\>\> broadcast $\lit{NotarVote}(B_v^{\lit{timeout}},\sigma_j,\bot,\bot)$ \label{line:notarize_after_safe_to_skip}
        \State \>\>\>\> $\textit{notarized}\gets\textit{notarized}\cup\{B_v^{\lit{timeout}}\}$ \label{line:special_end}
    \EndLoop
\EndFor
\end{algorithmic}
\end{algorithm}

\begin{algorithm}
\caption{\proc{ReconstructAndNotarize}$(v,B)$}\label{alg:reconstruct}
\begin{algorithmic}[1]
\State reconstruct payload for $B$
\If{reconstruction succeeds $\land$ payload valid} \label{line:rvalid}
  \If{$B\notin\textit{notarized}$}
    \State broadcast $\lit{NotarVote}(B,\sigma_j,f_j,\pi_j)$
    \State $\textit{notarized}\gets\textit{notarized}\cup\{B\}$
  \EndIf
\Else
  \If{$B_v^{\lit{timeout}}\notin\textit{notarized}$}
    \State broadcast $\lit{NotarVote}(B_v^{\lit{timeout}},\sigma_j,\bot,\bot)$
    \State $\textit{notarized}\gets\textit{notarized}\cup\{B_v^{\lit{timeout}}\}$
  \EndIf
\EndIf
\end{algorithmic}
\end{algorithm}

\section{Protocol Analysis}

We start by proving some helpful properties.

\begin{lemma}[Validity Property]
Suppose that a block $B$ for some slot $v$ is added to the complete block tree of some replica. If the leader for slot $v$ is honest, $B$ must have been proposed by that leader.
\end{lemma}

\begin{proof}
By the Quorum Size Property (see \Cref{sec:sigs}) for notarization certificates, at least $n-2f-p$ honest
replicas must have broadcast notarization shares for $B$.
Since we are assuming $n \ge 3f+2p+1$, it follows that $n-2f-p > 0$, so some honest replica $P$
must have broadcast a notarization share for $B$.
This could happen either at line~\ref{line:notarize_as_first_vote}
or at line~\ref{line:invoke_reconstruct_and_notarize}.
In the first case, $B$ must be the block that $P$ received as a proposal from the leader.
In the second case, since $P$ received $f+p+1$ first votes for $B$,
one of these must be a first vote for $B$ from some honest replica $Q$,
and so $B$ must be the block that $Q$ received as a proposal from the leader.
\end{proof}

\begin{lemma}[Completeness Property for Certificates]
\label{lemma:cert_comp}
If a certificate $X$ appears in the vote and certificate pool (so $X$ is a notarization, finalization, or timeout certificate) then $X$ (or its equivalent) will eventually appear in the corresponding pool of every other replica. Moreover, if $X$ appears in a replica’s pool at a time $T$ at which the network is $\delta$-synchronous, it will appear in every replica’s pool before time $T + \delta$.
\end{lemma}
\begin{proof}
This is clear, since a certificate appearing in the vote and certificate pool is broadcast immediately.
\end{proof}

\begin{lemma}[Completeness Property for Blocks]
\label{lemma:block_comp}
If a block $B$ appears in the complete block tree, then $B$ will eventually appear in the corresponding tree of every other replica. Moreover, if $B$ appears in a replica’s tree at a time $T$ at which the network is $\delta$-synchronous, it will appear in every replica’s tree before time $T + \delta$.
\end{lemma}
\begin{proof}
We are relying on the Quorum Size Property (see \Cref{sec:sigs}) for notarization certificates: when a notarization certificate for a block $B$ is added to the certificate pool, at least $n - 2f - p$ honest replicas must have already broadcast notarization votes for $B$, which contain $B$ as well as fragments sufficient to reconstruct $B$’s payload, since $n -2f -p \ge f+p+1$.
\end{proof}

\subsection{Safety}

\begin{lemma}[Fast Finalization Implication]
\label{lemma:fast_final_implication}
Suppose a block $B$ is fast finalized by some honest replica, then the number of honest replicas that first vote for anything other than $B$ is at most $p$.
\end{lemma}

\begin{proof}
Let $f'\le f$ be the actual number of corrupt replicas.
If some honest replica fast finalizes $B$, then --- 
by the Quorum Size Property (see \Cref{sec:sigs}) for fast finalization certificates ---
at least $n-p-f'$ honest replicas first voted for $B$.
So the number of honest replicas that first vote for anything other than $B$ is at most $(n-f')-(n-p-f') = p$.
\end{proof}

\begin{lemma}[Uniqueness of Fast Finalization Property]
\label{lemma:safe_to_notarize}
If an honest replica receives $f + p + 1$ first votes for a block $B$ in slot $v$, then no block different from $B$ can be fast finalized in slot $v$ by any honest replica.
\end{lemma}
\begin{proof}
Suppose towards contradiction that and some replica $P$ receives  $f+p+1$ first votes for a 
block $B$ but 
 a block $C\ne B$ is fast finalized by some honest replica $Q$.
By the previous lemma, at most $p$ honest replicas could first vote for anything other than $C$.
Therefore, $P$ can receive at most $f+p$ first votes for $B$, a contradiction.
\end{proof}

\begin{lemma}[Absence of Fast Finalization Property]
\label{lemma:safe_to_skip}
If the inequality $\proc{allVotes}(\textit{firstVote}) - \proc{maxVotes}(\textit{firstVote}) \geq f + p + 1$ holds for a replica in slot $v$, then no block can be fast finalized in slot $v$ by any replica. 
\end{lemma}
\begin{proof}
Suppose towards contradiction that a replica $P$ fast finalized block $B$ while for replica $Q$ the inequality holds. 
By \Cref{lemma:fast_final_implication}, at most $p$ honest replicas fast vote
for anything other than $B$. 
Therefore, $Q$ can receive at most $f+p$ fast votes for anything other than $B$.
Let $\textit{count}_B$ denote the number of first votes for $B$ received by $Q$.
It follows that, at any point in time, for replica $Q$, we have
\[
\proc{allVotes}(\textit{firstVote}) - \proc{maxVotes}(\textit{firstVote}) 
\le \proc{allVotes}(\textit{firstVote}) - \textit{count}_B \le f+p,
\]
a contradiction.
\end{proof}

\begin{lemma}[Incompatibility of Notarization and (Fast) Finalization Property]
\label{lemma:incompatibility_notarization_finalization}
Suppose that a valid block $B$ for some slot $v$ is (fast) finalized by some replica. If any replica obtains a notarization certificate for a block $C$ in slot $v$, then $C = B$. 
(In particular, no other block for slot $v$ can be added to the complete block tree of any replica and no timeout certificate can be obtained for slot $v$.)
\end{lemma}
\begin{proof}
We first prove the Incompatibility of Notarization and \textit{Fast} Finalization. Suppose towards contradiction that for slot $v$ a fast finalization certificate exists for block $B$ and a notarization certificate exists for block $C$, with $C\neq B$. By \Cref{lemma:fast_final_implication} at most $p$ honest replicas first vote for anything other than $B$. By the Quorum Size Property (see \Cref{sec:sigs}) at least $(n - f - p) - f \ge f + p + 1$ honest replicas must send a notarization vote for $C$.  Thus, at least $f+1$ honest replicas must send a notarization vote for $C$ \textit{after} having sent a first vote.

\begin{itemize}
    \item 
On the one hand, suppose that a notarization vote for $C$ was sent by the protocol on line~\ref{line:invoke_reconstruct_and_notarize}. 
Due to the condition on line~\ref{line:condition_many}, this means that $f + p + 1$ first votes were received for a block $D$. 
Note that $D$ is a non-timeout block.
Moreover, by the logic of $\proc{ReconstructAndNotarize}$, either $C$ is a timeout block or $D=C$.
By the Uniqueness of Fast Finalization Property (\Cref{lemma:safe_to_notarize}) we know that $D = B$. 
Since $B$ is assumed to be valid, $C$ cannot be a timeout block, so we also have $B=D=C$, a contradiction.
\item On the other hand, suppose that a notarization vote for $C$ was sent on line~\ref{line:notarize_after_safe_to_skip}. By the Absence of Fast Finalization Property (\Cref{lemma:safe_to_skip}), this implies that no block is fast finalized, again a contradiction.
\end{itemize}

The Incompatibility of Notarization and Finalization Property follows
from a standard quorum intersection argument, based on the fact that 
in each slot an honest replica issues a finalization vote only for a block only if it did not send a notarization vote for a different block in that slot (see line~\ref{line:condition_final_vote}).
Suppose towards contradiction that for slot $v$ a finalization certificate exists for block $B$ and a notarization certificate exists for block $C$, with $C\neq B$.
By the Quorum Size property (see \Cref{sec:sigs}) for finalization and notarization certificates, 
if $f'\le f$ is the number of corrupt replicas,
then  at least $n-f-p-f'$
honest replicas broadcast finalization votes for $B$,
and a disjoint set of at least the same number of honest replicas
broadcast notarization votes for $C$.
This implies that there are at least $2(n-f-p-f')$ distinct honest replicas.
However, under the assumption that $n \ge 3f+2p+1$, we have $2(n-f-p-f') \ge n-f'+1$,
a contradiction.
\end{proof}

We can now easily state and prove our main safety lemma:

\begin{lemma}[Safety]
\label{lemma:safety}
Suppose a replica $P$ explicitly finalizes a block $B$ for slot $v$, and a block $C$ for slot $w \geq v$ is in the complete block tree of some replica $Q$. Then $B$ is an ancestor of $C$ in $Q$’s complete block tree.
\end{lemma}
\begin{proof}
By the Incompatibility of Notarization and (Fast) Finalization Property (\Cref{lemma:incompatibility_notarization_finalization}), no timeout certificate for slot $v$ can be produced. Let $C'$ be the parent of $C$ and suppose $w'$ is the slot number of $C'$. Since $C'$ is in $Q$’s complete block tree, a notarization certificate for $C'$ must have been produced, which means at least one honest replica must have issued a notarization vote for $C'$, which means $v \leq w' < w$. The inequality $v \leq w'$ follows from the fact that there is no timeout certificate for slot $v$, and an honest replica will issue a notarization share for $C$ only if it has timeout certificates for slots $w' + 1, \cdots, w - 1$. If $v = w'$, we are done by the Incompatibility of Notarization and (Fast) Finalization Property (\Cref{lemma:incompatibility_notarization_finalization}), and if $v < w'$, we can repeat the argument inductively with $C'$ in place of $C$. 
\end{proof}

\subsection{Liveness}

Liveness follows from the following lemmas.
The first lemma analyzes the optimistic case
where the network is synchronous and the leader of a given slot
is honest, showing that the leader's block will be committed.

\begin{lemma}[Liveness I]
\label{lemma-live1}
Consider a slot $v \ge 1$ and suppose the leader 
for slot $v$ is an honest replica $Q$.
Suppose that the first honest replica $P$ to
enter the loop iteration for slot $v$ does so at time $T_0$. 
Suppose that the network is $\delta$-synchronous
over the interval $[T_0, T_0+4\delta]$ for some $\delta$ with $\Delta\RM{timeout} \ge 2\delta$.
Then, $Q$ will propose a block for slot $v$ by time $T \le T_0+\delta$. Each honest replica will finish the loop iteration before time
$T + 2\delta$ by adding $Q$'s proposed block $B$ to its complete block tree.
Moreover, if $n-p$ replicas are honest, each honest replica will finalize $B$ by time $T+2\delta$. If more than $p$ replicas are corrupt, each honest replica will finalize $B$ by time $T+3\delta$.
\end{lemma}

\begin{proof}
By the Completeness Properties (\Cref{lemma:cert_comp} and \Cref{lemma:block_comp}), before time $T_0+\delta$, 
each honest replica
will enter  slot $v$ by time $T_0+\delta$,
having either a timeout certificate for slot $v-1$  or 
a block for slot $v-1$ in its complete block tree.
Before time $T \le T_0+\delta$, the leader $Q$ will propose a block $B$
that extends a block $B'$ with slot number $v' < v$.
By the logic of the protocol, we know that $Q$ must have timeout certificates for slots $v'+1, \ldots, v-1$ at the time
it makes its proposal, as well as a notarization certificate for $B'$.
Again by the Completeness Properties, before time $T+\delta$,
each honest replica will have $B'$ in its
complete block tree and all of these timeout certificates
in its certificate pool. 
Each honest replica will receive $B$ before this time,
and because $\Delta\RM{timeout} \ge 2\delta$,
will broadcast a first vote for $B$ by this time. 
Because all honest replicas broadcast a first vote for $B$,
each such replica will only see ever see at most $f$ first votes for any other block.
It follows that each honest replica will only ever see
\begin{itemize}
\item 
 $\proc{manyVotes}(\textit{firstVote}) \subseteq \{B\}$, and
\item
\(
\proc{allVotes}(\textit{firstVote})-\proc{maxVotes}(\textit{firstVote}) \le 
\proc{allVotes}(\textit{firstVote})-\mathit{count}_B \le f,
\)
where $\mathit{count}_B$ is the number of of first votes for $B$ that it sees.
\end{itemize}
Therefore, 
honest replicas will not broadcast a  notarization votes in slot $v$ for anything other than $B$.
Before time $T+2\delta$, each honest replica will
have added $B$ to its complete block tree and broadcast a finalization vote on $B$. If $n-p$ replicas are honest, each honest replica will have added a fast finalization certificate to its pool by time $T+2\delta$ as well.
Otherwise, if more than $p$ replicas are corrupt, each honest replica will finalize $B$ before time $T+3\delta$, when adding the finalization certificate for slot $v$ to its pool.
\end{proof}

The second lemma analyzes the pessimistic case,
when the network is asynchronous or the leader of a given
slot is corrupt.
It says that eventually, all honest replicas will move on to the next slot.

\begin{lemma}[Liveness II]
\label{lemma-live2}
Suppose that the network is $\delta$-synchronous over an interval
$[T, T+\Delta\RM{timeout}+3\delta]$,
for an arbitrary value of $\delta$,
and that at time $T$, some honest replica is in the loop iteration
for slot $v$ and all other honest replicas are in a loop iteration
for $v$ or a previous slot.
Then, before time $T+\Delta\RM{timeout}+3\delta$,
all honest replicas exit slot $v$.
\end{lemma}

\begin{proof}
By the Completeness Properties (\Cref{lemma:cert_comp} and \Cref{lemma:block_comp}), every honest replica will enter 
the  slot $v$ before time $T+\delta$. 
By time $T+\delta+\Delta\RM{timeout}$, every honest replica will
broadcast a first vote either for a block proposal, or for $B_v^{\lit{timeout}}$.

Consider two cases:
\begin{enumerate}[(a)]
    \item At least $f+p+1$ honest replicas broadcast a first vote for the same non-timeout block $B$.\label{itemlive1}
    \item No set of $f+p+1$ honest replicas broadcast a first vote for the same non-timeout block $B$.\label{itemlive2}
\end{enumerate}

\medskip
\noindent
{\em Case (\ref{itemlive1}).} Since $f+p+1$ replicas cast notarization votes for $B$, some honest replica did so. Since this replica had to have $B$'s parent in its complete block tree, by Completeness Properties (\Cref{lemma:block_comp}) each honest replica will have $B$'s parent in its complete block tree by time $T+\Delta\RM{timeout}+2\delta$. 
All honest replicas observe all first votes from other honest replicas before time $T+\Delta\RM{timeout}+2\delta$, and so each honest replica will call $\proc{ReconstructAndNotarize}(v,B)$
before that time, unless it has already exited slot $v$.
If some honest replica has exited slot $v$ before that time, the by Completeness Properties 
(\Cref{lemma:cert_comp} and \Cref{lemma:block_comp}), all honest replicas
will exit the slot before time $T+\Delta\RM{timeout}+3\delta$.
Otherwise, assume no honest replica has exited before time $T+\Delta\RM{timeout}+2\delta$.
Whenever $\proc{ReconstructAndNotarize}(v,B)$ is called by some honest replica, it has received at least $f+p+1$ first votes for $B$, and  so can attempt to reconstruct $B$. 
If it fails to reconstruct $B$'s payload or finds that it is invalid, then it and all honest replicas will do so and issue a timeout vote (this follows from collision resistance of the hash function
and the fragment decoding logic).
Otherwise, each honest replica will issue a notarization vote for $B$.
Therefore, before time $T+\Delta\RM{timeout}+3\delta$, each honest replica will either add $B$ to its complete block tree or the timeout certificate to its pool, and proceed exit the slot.

\medskip
\noindent
{\em Case (\ref{itemlive2}).} Consider some honest replica $P$. By time $T+\Delta\RM{timeout}+2\delta$, $P$ will have observed all votes of other honest replicas. 
If the only entries in $\textit{firstVote}$ are those from honest replicas,
then 
the inequality $\proc{allVotes}(\textit{firstVote})-\proc{maxVotes}(\textit{firstVote})\ge f+p+1$ must  hold.
To see this, if $f' \le f$ is the number of corrupt replicas,
then 
\begin{align*}
\proc{allVotes}(\textit{firstVote})-\proc{maxVotes}(\textit{firstVote}) & \ge (n - f') -(f + p) \\ 
& \ge n - 2f - p \\
& \ge f+p+1 \quad \text{(since $n \ge 3f+2p+1$).}
\end{align*}
As we have already observed, the quantity $\proc{allVotes}(\textit{firstVote})-\proc{maxVotes}(\textit{firstVote})$ cannot decrease as we add entries to $\textit{firstVote}$.
Therefore 
if $P$ has not cast a timeout vote in slot $v$ yet, it will do so by time $T+\Delta\RM{timeout}+2\delta$,
unless it has already exited slot $v$ by that time.
In either case, 
all honest replicas will exit slot $v$ by time $T+\Delta\RM{timeout}+3\delta$.
\end{proof}

\subsection{Boundedness}
\label{sec:bounds}

We prove some simple results that allow us to bound message and storage complexity.
Here, we are assuming both (\ref{eq:n_lower}) and (\ref{eq:n_upper}).

Let us consider notarization votes 
on non-timeout blocks.
An honest replica sends a first vote for at most one such block. 
Any notarization vote for some other non-timeout block $B$  requires that the replica found $B \in \proc{manyVotes}(\mathit{firstVotes})$. 
In other words, the replica has seen at least $f+p+1$ other replica's first votes for $B$. 
At most one first vote per replica is considered when computing $\proc{manyVotes}(\mathit{firstVotes})$. 
Therefore, an honest replica can cast at most $\lfloor n/(f+p+1)\rfloor$ notarization votes beyond the first vote, and by (\ref{eq:n_upper}), $\lfloor n/(f+p+1)\rfloor \le 2$, for a total of 3 notarization votes
for non-timeout blocks.

Suppose there are $f' \le f$ corrupt replicas, and so $n-f'$ honest replicas.
The honest replicas therefore issue at most $3(n-f')$ notarization votes for non-timeout blocks
per slot.

Now, to construct a notarization certificate
for a non-timeout block $B$, 
we require that $(n-f-p-f')$ honest replicas cast a notarization vote for $B$.
Therefore, by the result in the previous paragraph, there can be at most
\begin{equation}
\label{eq:N}
    N \deq \lfloor (3(n-f'))/(n-f-p-f') \rfloor
\end{equation}
distinct blocks for which a notarization certificate can be constructed.

We claim that $N \le 5$.
To see this, first note that
the derivative of $(3(n-f'))/(n-f-p-f')$ with respect to $f'$
is positive, and therefore the right-hand side of (\ref{eq:N}) is maximized when $f=f'$,
and so
\[
    N \le \lfloor (3(n-f))/(n-2f-p) \rfloor .
\]
So it suffices to show that $(3(n-f))/(n-2f-p) < 6$.
This is easily seen to follow by a simple calculation using (\ref{eq:n_lower}).

So to summarize, we have shown that in any slot, 
\begin{enumerate}
\item 
each replica casts a notarization vote for at most $2$ non-timeout blocks besides its first vote, and
\item 
there are at most $5$ distinct blocks for which a notarization certificate can be constructed. 
\end{enumerate}
This immediately gives us bounds on the message and storage complexity of the protocol per slot.

\subsection{Complexity}

Based on the concrete bounds in \Cref{sec:bounds} (and the
preliminary discussion in \Cref{sec:pool}), it is easily seen that
the message complexity per slot is $O(n^2)$.
Based on the properties of erasure codes and Merkle trees discussed in \Cref{sec:avid},
the communication complexity per slot is $O(\beta n + n^2 \log(n) \kappa + n^2 \lambda)$,
where $\beta$ is a bound on the payload size,  $\kappa$ is the output length of
the collision-resistant hash, and $\lambda$ is a bound on the length of any
signature share or certificate.
Moreover, the communication is {\em balanced}, in that every replica, {\em including the leader},
transmits the same amount of data, up to a constant factor.
It is also easily seen that each replica needs to store $O(\beta  + n \log(n) \kappa + n \lambda)$
bits of data.

\section{Protocol Variations}

Our protocol can be adapted to the setting of $n \ge 3f + 2p^* -1$, where $p^* \ge 1$ ~\cite{abraham2021good}, with the insight that if honest replicas vote for different blocks in the same slot, the leader has to be corrupt. 
The liveness analysis can leverage the fact that, in this case, honest replicas will observe at least $n-f$ votes from non-leader replicas, as is done in~\cite{banyan}. We plan to analyze a variation of \prot\ with this adaptation in follow-up work.

A variation of DispersedSimplex~\cite{song} features segments of consecutive slots with the same leader. Such stable leader assignment is beneficial for the throughput of the protocol. The same technique can be applied to our protocol, and we plan to analyze a variation of \prot\ featuring stable leaders in follow-up work.

\bibliographystyle{splncs04}
\bibliography{refs.bib}

\end{document}